\documentclass[journal]{IEEEtran}
\IEEEoverridecommandlockouts
\usepackage{cite}
\usepackage{multirow}
\usepackage{amsmath,amssymb,amsfonts,mathtools,bm}
\usepackage{amsthm}
\usepackage{algorithm}
\usepackage{algorithmic}
\usepackage{graphicx}
\usepackage{textcomp}
\usepackage{xcolor,float}
\usepackage{colortbl}
\usepackage{tabularx}
\usepackage{textcomp}
\usepackage{diagbox}
\usepackage{svg}
\usepackage{booktabs}
\usepackage{hyperref}
\usepackage{pifont}

\newtheorem{problem}{Problem}
\newtheorem{assumption}{Assumption}
\usepackage{threeparttable}
\usepackage{booktabs}
\usepackage{tabularx}
\usepackage[utf8]{inputenc}
\usepackage[T1]{fontenc}
\usepackage{float}
\usepackage[labelformat=empty]{subcaption}
\usepackage{graphicx}
\usepackage{multirow} % 在导言区加入
\usepackage{siunitx}
\usepackage{placeins}

\usepackage{orcidlink}
\usepackage[final]{changes} %宏包
\newtheorem{proposition}{Proposition}

\allowdisplaybreaks
\renewcommand{\baselinestretch}{1}

\usepackage{orcidlink}  % 加到导言区

\begin{document}
\title{RadioDiff-v2: Generative Angular Radio Maps for Multi-Beam Selection and Localization}

% \author{Qiming Zhang\orcidlink{0009-0004-4048-4668},
%         Xiucheng Wang\orcidlink{0000-0001-2345-6789},~\IEEEmembership{Student Member,~IEEE,}
%         Nan Cheng\orcidlink{0000-0001-7907-2071},~\IEEEmembership{Senior Member,~IEEE,}}
%         Ruijin Sun,\orcidlink{}~\IEEEmembership{Member,~IEEE,}

\author{
Xiucheng Wang,
Junxi Huang,
Nan Cheng

\thanks{
\par This work was supported by the National Key Research and Development Program of China (2024YFB907500).
\par Xiucheng Wang, Junxi Huang and Nan Cheng are with the State Key Laboratory of ISN and School of Telecommunications Engineering, Xidian University, Xi’an 710071, China (e-mail: \{xcwang\_1, 24012100067\}@stu.xidian.edu.cn; dr.nan.cheng@ieee.org);\textit{(Corresponding author: Nan Cheng.)}.

}

} 

    \maketitle

\IEEEdisplaynontitleabstractindextext

\IEEEpeerreviewmaketitle

\begin{abstract}
Angular radio maps describe the received-power distribution over the angle of arrival and underpin beam selection and receiver localization in sixth-generation (6G) networks. Predicting the angular power spectrum (APS) from geometry is difficult, because the mapping is ill-posed in non-line-of-sight (NLOS) conditions and must generalize to unseen environments. Distortion-minimizing regressors return the conditional mean, which over-smooths the spectrum and erases the multipath structure that downstream tasks need. We cast the task as a perception-distortion problem and propose RadioDiff-v2, a dual-branch one-dimensional diffusion transformer trained with flow matching. It couples periodic angular encoding, adaptive layer-normalization conditioning, a Fourier angular mixer, and joint velocity and clean-signal heads. A per-metric estimator portfolio reads every deployment quantity from this single model, so that samples carry the distribution, the clean-signal head supplies a regression-grade point estimate, Bayes-optimal rules select beams, and the conditional likelihood localizes the receiver. We prove that a concentrated conditional yields a straight probability-flow trajectory that one step integrates exactly, identifying deterministic transport as the correct inductive bias. On a zero-shot test of 99 environments and one million links, RadioDiff-v2 leads every baseline on every metric, with a 0.39 dB Wasserstein-1 distance, per-bin error below the regression baseline, a 2.43 dB eight-beam NLOS sweep loss, and a 20.6-pixel localization error with four base stations. Code is available at \url{https://github.com/UNIC-Lab/RadioDiff-v2}.
\end{abstract}


\section{Introduction}
\label{sec:intro}

The angular structure of the wireless channel governs how a base station forms beams and where it places nulls. Sixth-generation (6G) systems push toward higher carriers and large antenna arrays, where directional transmission is the dominant source of link gain \cite{6g,docomo20165g}. A station that knows the angular power spectrum (APS) at a candidate receiver location can steer toward the strongest arrival before any pilot exchange. The same angular knowledge underpins beam management, interference avoidance, and device localization in dense cells. A radio map, also called a channel knowledge map, stores such location-dependent channel descriptors over space so that a station can query them offline \cite{li2022radionet,zhang2024generative,sun2024generative}. Most radio-map work targets a scalar received power or path loss \cite{levie2021radiounet,wang2024radiodiff}, whereas the angular radio map is the harder and the more useful object because it exposes direction rather than a single scalar. Beam selection and angle-based localization both consume the full angular profile rather than a single power value, and this paper therefore studies the prediction of the angular radio map from geometry.

This prediction is difficult because multipath propagation governs the mapping from geometry to the angular profile. In a line-of-sight (LOS) link the dominant arrival follows the geometric bearing from receiver to transmitter, so the angular profile is sharply peaked and nearly determined by the receiver position. In a non-line-of-sight (NLOS) link the direct path is blocked, and the received energy arrives through reflection and diffraction around buildings. The angular profile then carries several lobes whose directions depend on the surrounding layout in ways that no closed-form bearing rule captures. A predictor must also generalize to environments it never saw during training, since a deployed map cannot enumerate every city block. This zero-shot requirement rules out memorizing per-environment fingerprints and demands a model that reasons from the building geometry itself. Multipath ambiguity, NLOS ill-posedness, and zero-shot generalization together make the angular radio map an open challenge for learning-based channel modeling.

Existing learning approaches treat the problem as image-to-image regression and minimize a per-pixel reconstruction loss. A convolutional regressor such as RadioUNet trains a U-shaped network with a mean-squared-error objective and produces accurate scalar maps \cite{levie2021radiounet,ronneberger2015u}. A conditional generative adversarial network such as RME-GAN adds an adversarial term to a reconstruction loss to sharpen the outputs \cite{zhang2023rme,creswell2018generative}. The prior diffusion model RadioDiff casts map construction as denoising and reports strong path-loss accuracy with a convolutional backbone \cite{wang2024radiodiff,wang2025radiodiffk}, yet these methods share one assumption that fails for the angular task. A loss that scores each angular bin against the ground truth is minimized in expectation by the conditional mean of the angular profile. The conditional mean of several NLOS lobes is a broad, low-amplitude blur that sits between the true lobes and matches none of them, and this blur is the failure that the present work removes.

This averaging is not a tuning artifact but a structural property of distortion-minimizing training, and it defines the central difficulty we address. Angular radio-map prediction is therefore a perception-distortion problem \cite{blau2018perception}, in which a predictor can minimize per-bin distortion or reproduce the sharp multi-lobe statistics of the true angular profile, but in the NLOS regime it cannot do both. A regressor takes the first option and over-smooths, and the resulting blur collapses the dynamic range, merges distinct arrivals into one diffuse mass, and erases the angular separation between candidate beam directions. A beam chosen from a blurred profile then points at the average of the lobes rather than at any usable arrival, and a localizer that compares angular signatures loses the very sharpness that distinguishes nearby positions. The deterministic generative adversarial network suffers a related failure, because adversarial training on this task collapses onto a narrow set of outputs and does not represent the spread of plausible NLOS profiles. RadioDiff is generative, yet it samples a reverse-time stochastic differential equation that injects fresh noise at every step. This stochastic sampling is a poor match for a channel that is close to deterministic once the geometry is fixed. The task thus needs a generator that is faithful to the conditional distribution without injecting spurious noise.

We propose RadioDiff-v2, a generative model that matches the conditional distribution of the angular radio map rather than its conditional mean, implemented as a one-dimensional diffusion transformer trained with flow matching. It learns an ordinary differential equation that transports a noise prior to the data conditional along a near-straight path \cite{lipman2023flow,liu2023flow,peebles2023scalable}. This deterministic transport is the right inductive bias for a channel that is nearly determined by its environment, and it places probability on the sharp arrivals instead of spreading it. In a LOS link the learned map behaves like a sharp regressor without the averaging that destroys the dominant lobe. In an NLOS link it represents the several plausible arrivals, which yields diverse candidate beams for a sweep and a posterior mean for single-beam pointing, and the same conditional density also serves a second task that regressors cannot perform. Scoring an observed angular profile under each candidate position gives a Bayesian likelihood, so RadioDiff-v2 localizes a receiver by generative maximum-a-posteriori inference \cite{yapar2022locunet,wang2020indoor}. One model thus supplies sharp angular maps, calibrated beam diversity, an accurate point estimate, and a likelihood for localization, all zero-shot to unseen environments. The main contributions of this paper, spanning problem formulation, model design, theoretical analysis, and deployment, are summarized as follows.

\begin{enumerate}
  \item We reframe angular radio-map prediction as a perception-distortion problem. We show that a distortion-minimizing regressor returns the conditional mean of the angular profile. The mean over-smooths the spectrum, collapses its dynamic range, and erases the multipath structure that beam selection and likelihood-based localization require.
  \item We propose RadioDiff-v2, a dual-branch one-dimensional diffusion transformer trained with rectified flow. The design combines a periodic angular positional encoding, adaptive layer-normalization conditioning, a bottleneck adaptive Fourier transform angular mixer, and coupled velocity and clean-signal decoder heads with a learnable fusion gate. We further introduce a per-metric estimator portfolio that reads every deployment quantity from this single model. The sampling readout carries the distribution, the clean-signal head yields a regression-grade point estimate, and Bayes-optimal rules select single beams and diverse sweep sets. Every readout hyper-parameter is selected on validation data and frozen before testing.
  \item We analyze why flow matching fits the angular radio-map task. We prove that, for a concentrated conditional, the flow-matching probability-flow ordinary differential equation follows a straight-line trajectory that one step integrates exactly. A noise-injecting reverse-time stochastic differential equation instead retains residual spread. The result identifies deterministic transport as the correct inductive bias and connects it to the perception-distortion regime.
  \item We use the learned conditional density as a Bayesian likelihood to localize a receiver, a capability regressors cannot provide. On a zero-shot test of 99 environments and one million links, RadioDiff-v2 leads every baseline on every metric. It attains a distributional Wasserstein-1 distance of 0.39 dB against 1.97 for the prior diffusion baseline, a per-bin error of 0.184 against 0.199 for the regression baseline, and an eight-beam NLOS sweep loss of 2.43 dB. The multi-station localization error is 20.6 pixels at four base stations and falls to 17.1 pixels at six, where regressor-map fusion saturates near 41 pixels.
\end{enumerate}

\begin{figure*}[ht]
\centering
\includegraphics[width=0.88\textwidth]{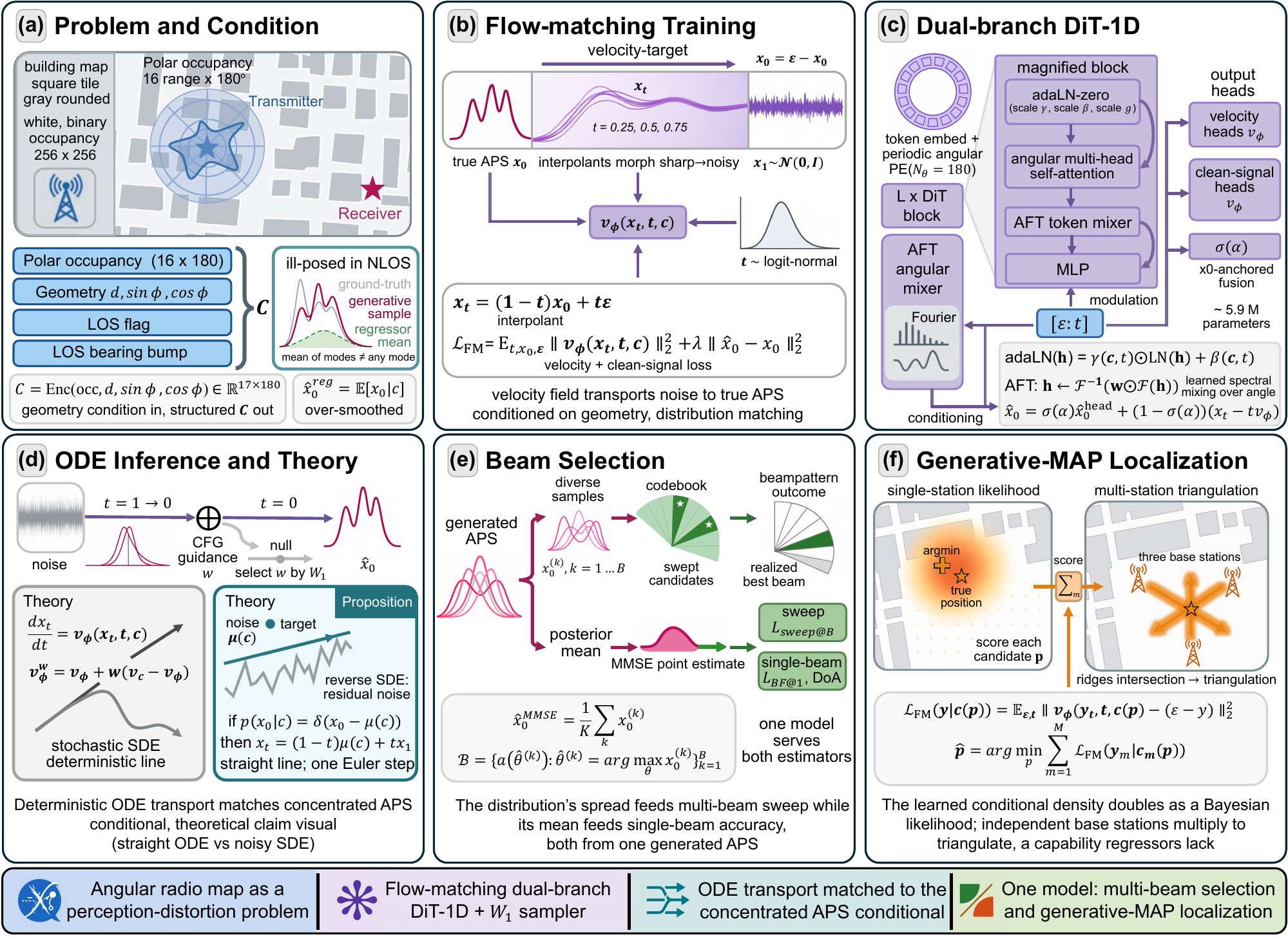}
\caption{Overview of RadioDiff-v2. (a)~The condition is built from the receiver-centred polar occupancy, the geometry, and the LOS flag, and the NLOS conditional is ill-posed for a mean estimator. (b)~Flow matching learns the velocity field that transports noise to the true spectrum under the condition. (c)~The dual-branch DiT-1D combines periodic angular encoding, adaptive layer-normalization conditioning, the AFT angular mixer, and velocity and clean-signal heads with a learnable fusion gate. (d)~The deterministic ODE transport matches the concentrated conditional, in contrast to a noise-injecting SDE. (e)~The estimator portfolio serves the beam sweep through diverse samples and single-beam pointing through the point estimate. (f)~The conditional likelihood localizes the receiver, and independent base stations triangulate.}
\label{fig:overview}
\end{figure*}

\section{Related Work}
\label{sec:related}

\subsection{Radio-Map and Channel-Knowledge-Map Construction}
Radio-map and channel-knowledge-map construction predicts a spatial field of channel quantities from the environment geometry and supports environment-aware operation in 6G networks~\cite{zeng2021toward,wang2026tutorial}. Early learning approaches cast path-loss-map prediction as image-to-image regression. RadioUNet trained a convolutional U-shaped network that maps a building layout and a transmitter location to a dense path-loss map under a mean-squared-error objective~\cite{levie2021radiounet}. The conditional generative adversarial network RME-GAN instead paired a reconstruction loss with an adversarial loss to sharpen the estimated field~\cite{zhang2023rme}. The RadioDiff family reframed the task as conditional generation, adopted a denoising diffusion probabilistic model with a convolutional backbone, and improved perceptual sharpness over the regressors~\cite{wang2024radiodiff}. Subsequent variants have extended this line in three directions. The first injected electromagnetic priors and inverse-problem structure into the diffusion process for multipath-aware and integrated sensing settings~\cite{wang2025radiodiffk,wang2025radiodiffinverse}. The second carried the diffusion prior to indoor construction-and-localization and to few-shot regimes through physics-informed manifold alignment~\cite{wang2025iradiodiff,wang2026radiodifffs}. The third reused diffusion-trajectory midpoints to cut inference latency~\cite{wang2026radiodiffflux}. A parallel effort has broadened coverage with graph-based reconstruction and large datasets that frame channel-knowledge maps as a computer-vision problem, including three-dimensional radio-map benchmarks across heights~\cite{li2024radiogat,jaensch2024radio,zhang2024generative,wu2024ckmimagenet,wang2026radiodiff3d}. These methods share a scalar path-loss target on a Cartesian grid, and they optimize per-pixel agreement with a ground-truth map. This per-pixel emphasis over-smooths the angular fine structure of the channel, which is exactly the structure that multi-beam selection and likelihood-based localization require. RadioDiff-v2 instead targets the dB APS and matches its conditional distribution, recovering the sharp multipath lobes that the prior estimators wash out.

\subsection{Generative Models}
Diffusion and flow-based generators define the methodological backdrop for distribution-faithful prediction. A denoising diffusion probabilistic model learns to invert a fixed forward noising chain and samples through a reverse-time stochastic process~\cite{ho2020denoising}. Deterministic and score-based formulations have recast sampling as an ordinary or stochastic differential equation and accelerated inference~\cite{song2020denoising,song2020score}. Latent and text-guided variants have scaled generation to high resolution and rich conditions~\cite{LDM,nichol2021glide}. Flow matching and rectified flow take a complementary route and regress a velocity field along a straight interpolant between the noise prior and the data. This construction yields a probability-flow ordinary differential equation with short, low-curvature trajectories~\cite{lipman2023flow,liu2023flow}. Generative adversarial networks reach sharp samples through a min-max game but are prone to mode collapse~\cite{creswell2018generative}. Decoupled diffusion designs have reshaped the transition kernel to ease few-step sampling~\cite{huang2023decoupled}. Transformer backbones supply scalable conditioning, and the diffusion transformer pairs attention with adaptive normalization for strong generative quality~\cite{vaswani2017attention,peebles2023scalable}. Broad vision surveys have documented this progression in detail~\cite{croitoru2023diffusion}. These generators are designed for high-entropy image priors, whereas the APS is a near-deterministic physical quantity given the geometry. RadioDiff-v2 therefore adopts flow-matching ordinary-differential-equation transport on a one-dimensional diffusion transformer, which fits this concentrated conditional and avoids the residual stochastic spread of noise-injecting samplers.

\subsection{Radio-Map Localization and Fingerprinting}
Localization from radio observations is a long-standing application of radio maps. Fingerprinting matches a measured signature against a stored database and returns the position of the nearest entry. A Gaussian-process map often interpolates the radio field between the stored entries~\cite{wang2020indoor}. Deep variants learn the matching directly. LocUNet regressed a receiver position from path-loss radio maps with a convolutional network~\cite{yapar2022locunet}. Further studies have mined radio fingerprints for simultaneous localization and mapping or for robust device identification~\cite{liu2023exploiting,li2022radionet}. These approaches share a discriminative or nearest-neighbour view. They compare an observed signature to references under a fixed metric, and they cannot evaluate how probable an observation is at a candidate position. A regressor therefore offers no conditional likelihood, and a fingerprint match degrades once the observed angular structure is smoothed away. RadioDiff-v2 closes this gap by scoring each candidate position with the flow-matching loss of the observed spectrum under that position's condition. This scoring turns the learned conditional density into a Bayesian likelihood for generative maximum-a-posteriori localization~\cite{kang2025confidence}. Per-station likelihoods multiply across base stations, so the estimate triangulates as coverage grows and surpasses even an oracle fingerprint built from the true spectra.

\section{System Model and Problem Formulation}
\label{sec:sysmodel}

We consider a single-frequency wireless environment served by multiple base stations and described by a binary building map. The map is a $256\times256$ image whose pixels mark free space or building interior, and all coordinates are expressed in pixels over $[0,256]^2$. Each environment contains on the order of one hundred transmitter sites and tens of thousands of receiver locations. Propagation between any transmitter and any receiver is governed by LOS paths, reflections, and diffractions around the buildings. For a fixed transmitter and receiver, the directional structure of the received field is summarized by the APS, which records how the arriving power is distributed over the azimuth at the receiver. The APS encodes the directions of the dominant multipath components, so beam selection and angle-based localization both depend on it. This section defines the APS, the conditioning information available to the model, the zero-shot evaluation setting, and the two downstream tasks. Each symbol is introduced where it first appears, and the principal symbols are collected in Table~\ref{tab:notation}.

The generation target is the dB-domain APS at a receiver, written as $\bm{x}_0\in\mathbb{R}^{N_\theta}$. The azimuth is discretized into $N_\theta=180$ angular bins of one degree, and the $n$-th bin is centered at angle $\theta_n$ for $n=1,\dots,N_\theta$. Let $P(\theta_n)$ denote the linear received power arriving within the $n$-th bin, obtained by aggregating the multipath contributions whose angle of arrival falls in that bin. The $n$-th entry of the target is the normalized power in decibels,
\begin{align}
  [\bm{x}_0]_n = \frac{10\log_{10}\!\big(P(\theta_n)+\varepsilon\big) - P_{\min}}{P_{\max}-P_{\min}},
  \label{eq:aps_def}
\end{align}
where $\varepsilon$ is a small constant that bounds the logarithm in deep nulls, and $P_{\min}$ and $P_{\max}$ are fixed dynamic-range limits applied across the dataset. The angular power varies over several orders of magnitude, so the dB domain is adopted to expose the weak off-axis lobes that a linear scale would suppress. The mapping in~\eqref{eq:aps_def} is monotone, so the angle of the dominant lobe and the angular spread of $\bm{x}_0$ are both preserved. The APS is periodic in the azimuth, and the angular positional encoding introduced in Sec.~\ref{sec:method} respects this periodicity.

\begin{table}[ht]
\centering
\renewcommand{\arraystretch}{1.3}
\caption{Principal notation.}
\label{tab:notation}
\resizebox{\columnwidth}{!}{
\begin{tabular}{@{}l|l@{}}
\hline
Symbol & Meaning \\
\hline
$\bm{x}_0\in\mathbb{R}^{N_\theta}$ & Clean dB-domain angular power spectrum (target) \\
$N_\theta=180$ & Number of angular bins over the azimuth \\
$\theta_n$ & Azimuth angle of the $n$-th bin \\
$\bm{c}$ & Conditioning input (all items below jointly) \\
$\bm{B}\in\{0,1\}^{16\times180}$ & Rx-centred polar building-occupancy map \\
$\bm{p}_{\mathrm{tx}},\bm{p}_{\mathrm{rx}}\in[0,256]^2$ & Transmitter and receiver pixel coordinates \\
$\bm{g}$ & Tx/Rx geometry features \\
$\ell\in\{0,1\}$ & Line-of-sight flag; LOS angle withheld \\
$p(\bm{x}_0\mid\bm{c})$ & Conditional law of the APS given $\bm{c}$ \\
$\bm{\mu}(\bm{c})=\mathbb{E}[\bm{x}_0\mid\bm{c}]$ & Conditional mean (the regressor target) \\
$\bm{\epsilon}\sim\mathcal{N}(\bm{0},\bm{I})$ & Gaussian noise prior sample \\
$t\in[0,1]$ & Flow-matching (rectified-flow) time \\
$\bm{x}_t$ & Rectified-flow linear interpolant \\
$\bm{v}=\bm{\epsilon}-\bm{x}_0$ & Target velocity field \\
$\bm{v}_\phi(\bm{x}_t,t,\bm{c})$ & Network velocity prediction (parameters $\phi$) \\
$\hat{\bm{x}}_0$ & Network clean-signal prediction \\
$\sigma(\alpha)$ & Learnable $\bm{x}_0$-fusion gate \\
$w$ & Classifier-free guidance scale; null token $\varnothing$ \\
$\mathcal{W}_1(\cdot,\cdot)$ & Wasserstein-1 distance (dB domain) \\
$\bm{a}(\theta)$ & Array steering vector at angle $\theta$ \\
$\mathcal{B}=\{\bm{b}_1,\dots,\bm{b}_B\}$ & Set of $B$ candidate beams \\
$L_{\mathrm{sweep}@B}$ & Oracle-best minus realized beamforming gain (dB) \\
$G(\theta)$ & Beamforming gain at bearing $\theta$ (dB) \\
$\bm{p}^\star$ & Estimated receiver position \\
$\mathcal{L}_{\mathrm{FM}}(\bm{y}\mid\bm{c})$ & Flow-matching score of $\bm{y}$ under $\bm{c}$ \\
$M$ & Number of base stations heard \\
\hline
\end{tabular}}
\end{table}

The conditioning input collects everything the model is allowed to observe and is denoted by $\bm{c}$. Its first component is a receiver-centered, angle-aligned polar building-occupancy map $\bm{B}\in\{0,1\}^{16\times180}$, whose rays are cast from the receiver and record building occupancy over $16$ range cells and $180$ angular bins. The polar map shares its angular axis with the APS, so a blocked bearing in $\bm{B}$ aligns with the corresponding angular bin of $\bm{x}_0$. The second component is the geometry, namely the transmitter and receiver pixel coordinates $\bm{p}_{\mathrm{tx}},\bm{p}_{\mathrm{rx}}\in[0,256]^2$ together with derived features $\bm{g}$ that include the relative position and the transmitter-receiver distance. The third component is a binary LOS flag $\ell\in\{0,1\}$ that states whether an unobstructed path exists. The LOS bearing itself is withheld from $\bm{c}$, because revealing it would expose the dominant direction in LOS links and weaken the evaluation. The condition therefore carries occupancy, coarse geometry, and a visibility bit, but never the target angle.

The mapping from the condition to the APS is not deterministic under the observed information. Given the exact scatterer positions and material properties, the APS would be fixed by electromagnetic propagation~\cite{zhou2017electromagnetic,deschamps1972ray}. The condition $\bm{c}$ exposes only a binary occupancy map and coarse geometry, so several physical environments are consistent with the same $\bm{c}$, and the APS therefore retains a small but genuine residual uncertainty. We model this residual uncertainty by treating $\bm{x}_0$ as a sample from a conditional law $p(\bm{x}_0\mid\bm{c})$ that is concentrated for LOS links and admits limited multi-modality for NLOS links. The conditional mean is $\bm{\mu}(\bm{c})=\mathbb{E}[\bm{x}_0\mid\bm{c}]$, and a distortion-minimizing regressor learns to output exactly this quantity~\cite{levie2021radiounet}. Recovering the full conditional law rather than its mean is the central requirement of the angular radio-map task. The multipath structure that beam selection and localization need lives in the modes of $p(\bm{x}_0\mid\bm{c})$, and averaging over the modes erases it. The proposed method therefore targets $p(\bm{x}_0\mid\bm{c})$ directly through generative modeling.

Generalization is evaluated zero-shot across environments, and the $99$ environments~\cite{huang2026map2aps} are partitioned into $79$ disjoint training environments and $20$ disjoint test environments, so that no test building map appears during training~\cite{alkhateeb2019deepmimo}. Training draws links from the training environments, an in-distribution validation set is held out for monitoring, and the reported test set consists only of links from the unseen test environments. This split measures whether a model learns transferable propagation structure rather than memorizing the training geometries. A subset of receivers is heard by three or more transmitters, which gives a multi-base-station setting and enables the fusion used for localization. The zero-shot protocol is therefore the operating regime of interest, because a deployable radio-map model must serve environments for which no measurements were collected, and all quantities defined below are evaluated under this protocol.

The first downstream task is beam selection in the NLOS regime. In this regime the direct path is blocked, and the optimal beam points along a reflected or diffracted bearing rather than the geometric transmitter direction~\cite{wahl2005dominant}. A beam is a steering vector $\bm{a}(\theta)$ from an angular codebook, and its realized beamforming gain is the APS power collected along $\theta$. A generative model proposes a set of $B$ candidate beams $\mathcal{B}=\{\bm{b}_1,\dots,\bm{b}_B\}$, and the receiver sweeps them. The relevant loss is therefore the gap between the oracle-best beam and the best realized beam over the set,
\begin{align}
  L_{\mathrm{sweep}@B}(\mathcal{B}) = \max_{\theta} G(\theta) - \max_{\bm{b}\in\mathcal{B}} G(\bm{b}),
  \label{eq:sweep_loss}
\end{align}
where $G(\theta)$ denotes the beamforming gain at bearing $\theta$ under the true APS, measured in decibels. We formalize the task as follows.
\begin{problem}[Generative beam sweep]
\label{prob:beam}
Given the condition $\bm{c}$, produce $B$ candidate beams that minimize the expected sweep loss,
\begin{align}
  \min_{\mathcal{B}}\ \mathbb{E}\big[L_{\mathrm{sweep}@B}(\mathcal{B})\,\big|\,\bm{c}\big],
  \quad |\mathcal{B}|=B.
  \label{eq:beam_problem}
\end{align}
\end{problem}
A small sweep loss requires candidate beams that are both accurate and diverse, so a model that collapses to a single bearing wins only when $B=1$. The angular spread of the predicted APS is reported alongside~\eqref{eq:sweep_loss}, because preserving the spread supplies the diversity that the sweep rewards~\cite{zhang2023rme}.

The second downstream task is receiver localization from the observed APS, which we cast as a Bayesian estimation problem over candidate positions. Let $\bm{y}$ denote an observed APS measured at an unknown receiver. Let $\bm{c}(\bm{p})$ denote the condition that the known building map induces for a candidate position $\bm{p}$. A generative model assigns each candidate the conditional likelihood of the observation under that position. We summarize this likelihood by the flow-matching score $\mathcal{L}_{\mathrm{FM}}(\bm{y}\mid\bm{c}(\bm{p}))$ defined in Sec.~\ref{sec:method}. The estimate is the position whose induced condition best explains the observation,
\begin{align}
  \bm{p}^\star = \arg\min_{\bm{p}}\ \mathcal{L}_{\mathrm{FM}}\big(\bm{y}\,\big|\,\bm{c}(\bm{p})\big).
  \label{eq:loc_map}
\end{align}
When the receiver is heard by $M$ base stations, the per-station likelihoods are conditionally independent given the position and therefore combine additively in the log domain. We state the task as follows.
\begin{problem}[Generative maximum-a-posteriori localization]
\label{prob:loc}
Given $M$ observed spectra $\{\bm{y}_m\}_{m=1}^{M}$ and the building map, estimate the receiver position by
\begin{align}
  \bm{p}^\star = \arg\min_{\bm{p}}\ \sum_{m=1}^{M}\mathcal{L}_{\mathrm{FM}}\big(\bm{y}_m\,\big|\,\bm{c}_m(\bm{p})\big),
  \label{eq:loc_problem_multi}
\end{align}
where $\bm{c}_m(\bm{p})$ is the condition induced for the $m$-th base station.
\end{problem}
Each station contributes a likelihood ridge along its bearing, and the ridges intersect at the true position, so the estimate sharpens as $M$ grows. A regressor returns only the conditional mean and provides no likelihood, so it cannot score~\eqref{eq:loc_map} and cannot perform this task~\cite{yapar2022locunet}. Problems~\ref{prob:beam} and~\ref{prob:loc} therefore both reduce to modeling the conditional law $p(\bm{x}_0\mid\bm{c})$, which is the objective of the method developed next.

\section{Proposed Method: RadioDiff-v2}
\label{sec:method}

RadioDiff-v2 casts APS prediction as conditional generation rather than regression, and the model receives the condition $\bm{c}$ and learns to draw samples from the conditional law $p(\bm{x}_0 \mid \bm{c})$, where $\bm{x}_0 \in \mathbb{R}^{N_\theta}$ is the dB-domain APS over $N_\theta=180$ angular bins. The generator is a dual-branch one-dimensional diffusion transformer trained by flow matching \cite{lipman2023flow}, and Fig.~\ref{fig:overview} shows the full pipeline. The condition is encoded once and injected into every transformer block, and the network predicts both a velocity field and a clean signal. An ordinary differential equation (ODE) sampler then integrates the velocity from a Gaussian prior to an APS sample. Because the conditional density is modeled rather than its mean, one trained network serves three needs at once. It supplies a posterior mean for point estimation, a set of diverse samples for beam selection, and a conditional likelihood for localization. This section develops the training objective, the backbone, the distribution-aware sampler selection, the theoretical motivation, and the two deployment estimators in turn.

\subsection{Flow-Matching Formulation}
\label{sec:method_fm}

Flow matching learns a time-dependent velocity field whose probability-flow ODE transports a Gaussian prior to the data conditional \cite{lipman2023flow}. We adopt the rectified-flow interpolant, which connects a clean APS $\bm{x}_0$ and a noise sample $\bm{\epsilon}\sim\mathcal{N}(\bm{0},\bm{I})$ along a straight segment \cite{liu2023flow}. For a flow time $t\in[0,1]$, the interpolant and its target velocity are
\begin{align}
  \bm{x}_t &= (1-t)\,\bm{x}_0 + t\,\bm{\epsilon},
  \label{eq:interpolant} \\
  \bm{v} &= \frac{\mathrm{d}\bm{x}_t}{\mathrm{d}t} = \bm{\epsilon} - \bm{x}_0.
  \label{eq:target_velocity}
\end{align}
The target velocity in~\eqref{eq:target_velocity} is constant along each interpolation path. The network $\bm{v}_\phi(\bm{x}_t,t,\bm{c})$ with parameters $\phi$ regresses this velocity from the noisy state, the time, and the condition. The training objective is the conditional flow-matching loss
\begin{align}
  \mathcal{L}_{\mathrm{FM}}(\phi) = \mathbb{E}_{\bm{x}_0,\bm{\epsilon},t,\bm{c}}
  \big\lVert \bm{v}_\phi(\bm{x}_t,t,\bm{c}) - (\bm{\epsilon}-\bm{x}_0) \big\rVert_2^2,
  \label{eq:fm_loss}
\end{align}
with $t$ drawn uniformly on $[0,1]$ and $\bm{x}_t$ formed by~\eqref{eq:interpolant}. At inference, samples follow the probability-flow ODE
\begin{align}
  \frac{\mathrm{d}\bm{x}_t}{\mathrm{d}t} = \bm{v}_\phi(\bm{x}_t,t,\bm{c}),
  \qquad \bm{x}_1 = \bm{\epsilon},
  \label{eq:pf_ode}
\end{align}
integrated backward from $t=1$ to $t=0$. The regression target in~\eqref{eq:fm_loss} differs fundamentally from a per-bin reconstruction loss on $\bm{x}_0$. A reconstruction loss collapses the conditional onto its mean $\bm{\mu}(\bm{c})$, whereas~\eqref{eq:fm_loss} fits a transport map that preserves the conditional spread. This distinction is the source of the perception-distortion behavior analyzed in Sec.~\ref{sec:method_theory}.

\subsection{Dual-Branch DiT-1D Backbone}
\label{sec:method_backbone}

The backbone is a stack of one-dimensional transformer blocks operating over the $N_\theta$ angular tokens \cite{peebles2023scalable}, and each block applies multi-head self-attention and a feed-forward network with residual connections \cite{vaswani2017attention}. On top of this stack the design adds five components, and each is justified by the failure it prevents. The components are a periodic angular encoding, an adaptive normalization condition path, a Fourier angular mixer, coupled decoder heads, and classifier-free guidance, and we describe them in turn before assembling them into the sampler.

The first component is a periodic angular positional encoding, because the azimuth axis is circular, so bin $1$ and bin $N_\theta$ are neighbors. A standard linear positional encoding places a discontinuity at the wrap-around, which fragments lobes that straddle the $0$-degree boundary. We therefore encode each bin index $n$ with sinusoids of its angle $\theta_n$, so the encoding stays continuous across the seam. Without this periodic encoding, the model splits a single physical lobe into two partial lobes at the array boundary, and the dominant direction of arrival is mislocated.

The second component injects the condition through adaptive layer normalization in the adaLN-zero form \cite{peebles2023scalable}, where the condition $\bm{c}$ collects the polar building-occupancy map $\bm{B}$, the geometry features $\bm{g}$, and the LOS flag $\ell$. A condition encoder maps this input to per-block scale and shift parameters, and each block modulates its normalized activations by these parameters, so the geometry steers generation at every depth. The residual branches are initialized to zero, which makes each block an identity map at the start of training and stabilizes optimization. Concatenating the condition to the input instead would dilute it across depth, and the geometry signal would weaken in the deeper blocks where the lobe structure is resolved.

The third component is an adaptive Fourier transform (AFT) angular mixer at the network bottleneck. The APS is naturally sparse in the angular frequency domain, since a small number of multipath clusters produces a few dominant lobes. The AFT mixer transforms the bottleneck tokens to the angular spectral domain, applies a learned complex filter, and transforms back. This gives the model a global angular receptive field at low cost and a direct handle on the sparse spectral structure. Without the AFT mixer, the purely local attention smooths fine angular detail, and the recovered lobes are wider than the true ones.

The fourth component is a pair of decoder heads with a learnable fusion gate. One head predicts the velocity $\bm{v}_\phi$, and the other predicts the clean signal $\hat{\bm{x}}_0$ directly. The two predictions are tied through the interpolant, since~\eqref{eq:interpolant} and~\eqref{eq:target_velocity} give a velocity-implied clean signal $\hat{\bm{x}}_0^{v} = \bm{x}_t - t\,\bm{v}_\phi$. A learnable gate $\sigma(\alpha)\in(0,1)$ fuses the two estimates,
\begin{align}
  \hat{\bm{x}}_0^{\mathrm{fuse}} = \sigma(\alpha)\,\hat{\bm{x}}_0
  + \big(1-\sigma(\alpha)\big)\,\hat{\bm{x}}_0^{v},
  \label{eq:x0_fusion}
\end{align}
where $\sigma(\cdot)$ is the logistic function. The velocity head carries the transport that fits the distribution, while the direct $\hat{\bm{x}}_0$ head anchors the sample to a physically valid spectrum near $t=0$. The gate lets the model trust the direct head where it is reliable and the velocity head elsewhere. With the velocity head alone, the few-step trajectory accumulates error near the data end and the sharp peak is blunted.

The fifth component is classifier-free guidance with a learned null condition \cite{ho2022classifier}. During training, the condition is replaced by a learned null token $\varnothing$ with a fixed probability, so the network learns both the conditional and the unconditional velocity. At inference, the guided velocity extrapolates between the two,
\begin{align}
  \tilde{\bm{v}}_\phi(\bm{x}_t,t,\bm{c})
  = \bm{v}_\phi(\bm{x}_t,t,\varnothing)
  + w\big[\bm{v}_\phi(\bm{x}_t,t,\bm{c}) - \bm{v}_\phi(\bm{x}_t,t,\varnothing)\big],
  \label{eq:cfg}
\end{align}
with a guidance scale $w$. A larger $w$ sharpens the samples toward the conditional mode and trades sample diversity for fidelity to the dominant lobe. The role of $w$ is examined in Sec.~\ref{sec:method_sampler}.

These five components feed a sampler that integrates the guided probability-flow ODE in $K$ steps and anchors each step on the clean-signal estimate. At step $k$ with time $t_k$, the update reads the fused clean signal from~\eqref{eq:x0_fusion}, recomputes the guided velocity from~\eqref{eq:cfg}, and advances the state toward $t=0$. Anchoring on $\hat{\bm{x}}_0^{\mathrm{fuse}}$ keeps every intermediate state on the manifold of valid spectra, which matters under the small step budgets used here, and Algorithm~\ref{alg:sampler} summarizes the full procedure. A plain velocity-only Euler sampler drifts off the spectrum manifold within a few steps, which widens the lobes and inflates the distributional error.

\begin{algorithm}[ht]
\caption{$\bm{x}_0$-anchored ODE sampler for RadioDiff-v2.}
\label{alg:sampler}
\begin{algorithmic}[1]
\REQUIRE Condition $\bm{c}$, steps $K$, guidance scale $w$, schedule $1=t_0>\dots>t_K=0$
\ENSURE APS sample $\hat{\bm{x}}_0$
\STATE Draw $\bm{x}_{t_0}\sim\mathcal{N}(\bm{0},\bm{I})$
\FOR{$k=0,1,\dots,K-1$}
  \STATE Compute heads $\bm{v}_\phi$ and $\hat{\bm{x}}_0$ at $(\bm{x}_{t_k},t_k,\bm{c})$ and at $(\bm{x}_{t_k},t_k,\varnothing)$
  \STATE Form the guided velocity $\tilde{\bm{v}}_\phi$ by~\eqref{eq:cfg}
  \STATE Form the fused clean signal $\hat{\bm{x}}_0^{\mathrm{fuse}}$ by~\eqref{eq:x0_fusion}
  \STATE Set $\bm{x}_{t_{k+1}} \gets t_{k+1}\,\big(\bm{x}_{t_k}+(t_{k+1}-t_k)\,\tilde{\bm{v}}_\phi\big)/t_k + (1-t_{k+1}/t_k)\,\hat{\bm{x}}_0^{\mathrm{fuse}}$ \quad\COMMENT{$\bm{x}_0$-anchored step}
  \ENDFOR
\STATE $\hat{\bm{x}}_0 \gets \hat{\bm{x}}_0^{\mathrm{fuse}}$
\STATE \textbf{return} $\hat{\bm{x}}_0$
\end{algorithmic}
\end{algorithm}

\subsection{Distribution-Aware Sampler Selection}
\label{sec:method_sampler}

The guidance scale $w$ is an inference-time control, so it is selected after training without retraining. The usual practice tunes $w$ to minimize a per-bin distortion such as normalized mean-squared error (NMSE). This choice is misaligned with the angular tasks, because a distortion-optimal $w$ pushes the sampler toward the conditional mean and erases the multipath spread. We instead select $w$ by the Wasserstein-1 distance $\mathcal{W}_1$ between the distribution of generated APS values and the distribution of true APS values in the dB domain \cite{peyre2019computational}. The selected scale is
\begin{align}
  w^\star = \arg\min_{w}\; \mathcal{W}_1\big(p_w(\bm{x}_0),\, p(\bm{x}_0)\big),
  \label{eq:w1_selection}
\end{align}
where $p_w$ is the sampling law under guidance scale $w$. This criterion targets distributional fidelity directly. It restores the dynamic range and the lobe sharpness that the NMSE-optimal scale destroys, and it does so purely at inference. The trade-off governed by $w$ is quantified in Sec.~\ref{sec:experiments}.

\subsection{Theoretical Analysis}
\label{sec:method_theory}

The choice of an ODE transport over a noise-injecting reverse process is not incidental, because given the transmitter, the receiver, and the environment geometry, the APS is a near-deterministic physical quantity, so the conditional $p(\bm{x}_0 \mid \bm{c})$ is concentrated. We make this concentration precise and show that flow matching is the matched inductive bias. The argument rests on two assumptions, the first of which fixes the conditional to a deterministic limit while the second grants the optimal velocity.

\begin{assumption}[Concentrated conditional]
\label{assump:concentrated}
The conditional law of the APS given the condition is a Dirac mass at a deterministic map, $p(\bm{x}_0 \mid \bm{c}) = \delta(\bm{x}_0 - \bm{\mu}(\bm{c}))$. This is the exact limit of low aleatoric entropy and is approached most closely in the LOS regime.
\end{assumption}

\begin{assumption}[Optimal velocity]
\label{assump:optimal_velocity}
The network attains the minimizer of the flow-matching objective in~\eqref{eq:fm_loss}, so $\bm{v}_\phi$ equals the conditional expectation of the target velocity, $\bm{v}_\phi(\bm{x}_t,t,\bm{c}) = \mathbb{E}[\bm{\epsilon}-\bm{x}_0 \mid \bm{x}_t,\bm{c}]$.
\end{assumption}

\begin{proposition}[Straight-line transport for concentrated conditionals]
\label{prop:straight_line}
Under Assumptions~\ref{assump:concentrated} and~\ref{assump:optimal_velocity}, the flow-matching probability-flow ODE in~\eqref{eq:pf_ode} has the closed-form trajectory
\begin{align}
  \bm{x}_t = (1-t)\,\bm{\mu}(\bm{c}) + t\,\bm{x}_1,
  \label{eq:straight_line}
\end{align}
which is a straight line from the prior sample $\bm{x}_1$ to $\bm{\mu}(\bm{c})$. A single Euler step from $t=1$ to $t=0$ integrates~\eqref{eq:pf_ode} exactly, so the discretization error is zero. A reverse-time stochastic differential equation (SDE) with the same marginals instead injects fresh Gaussian noise of fixed per-step variance and does not share this property.
\end{proposition}
\begin{figure*}[ht]
\centering
\includegraphics[width=0.88\textwidth]{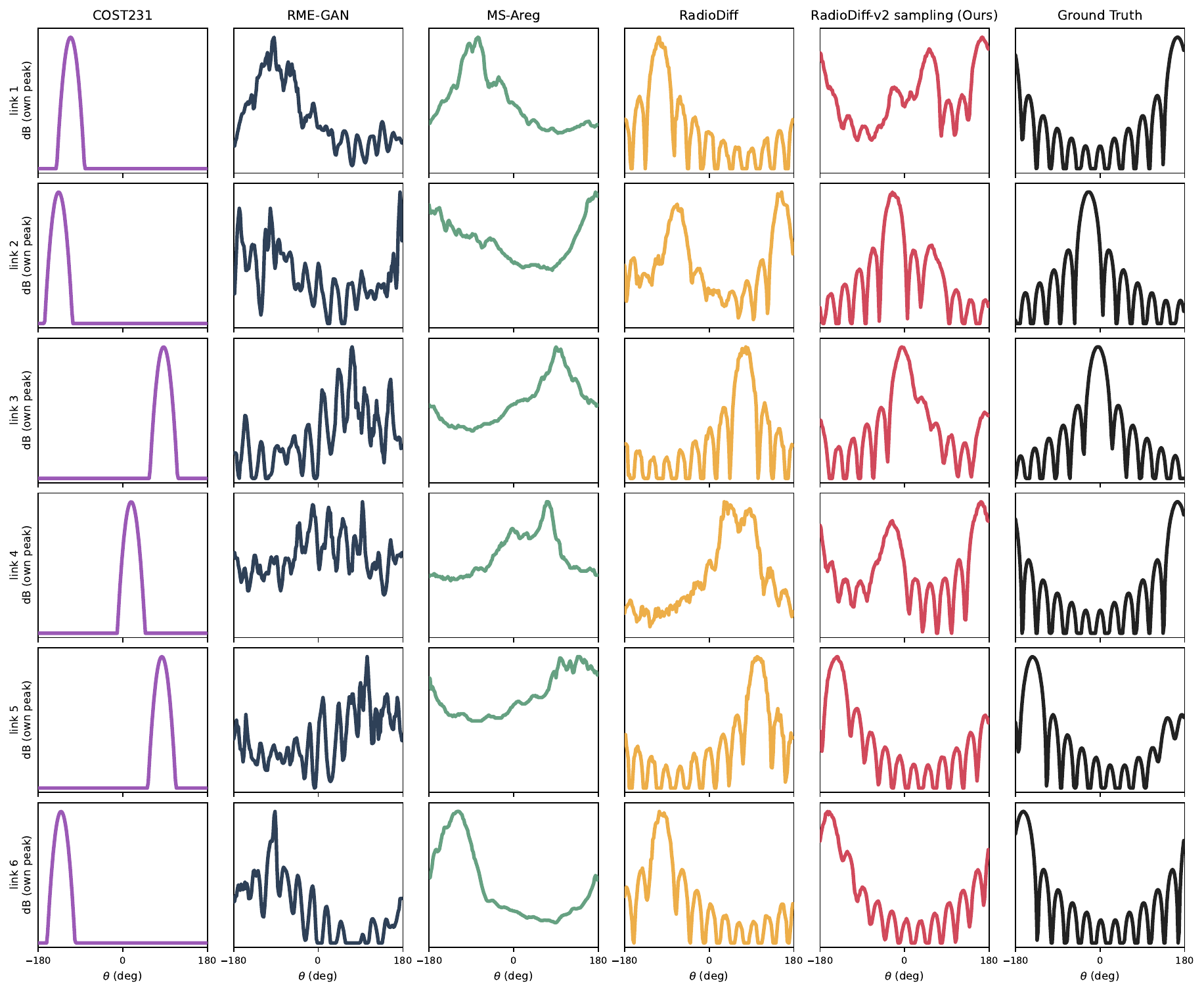}
\caption{Qualitative comparison of predicted angular power spectra on NLOS links, shown in the dB domain with each curve normalized to its own peak, so 0 dB marks that curve's dominant direction regardless of method. Columns from left to right are COST231, RME-GAN, MS-Areg, RadioDiff, the RadioDiff-v2 sampling readout, and the ground truth. The COST231 beam is steered to the geometric bearing, which under NLOS conditions points at the blocked direct path and lands at the wrong angle. The RadioDiff-v2 sampling readout consistently places its peak at the correct angle, whereas RME-GAN and MS-Areg stay broad and noisy near their own peak with no clean dominant lobe even after this per-curve referencing; the per-bin accuracy of the RadioDiff-v2 point readout is reported in Table~\ref{tab:main}.}
\label{fig:aps_qualitative}
\end{figure*}
\begin{IEEEproof}
Under Assumption~\ref{assump:concentrated} the only clean signal consistent with the condition is $\bm{x}_0=\bm{\mu}(\bm{c})$. Substituting into~\eqref{eq:interpolant} gives $\bm{x}_t=(1-t)\bm{\mu}(\bm{c})+t\bm{\epsilon}$, hence $\bm{x}_t$ is an affine function of $t$ and~\eqref{eq:straight_line} holds with $\bm{x}_1=\bm{\epsilon}$. Under Assumption~\ref{assump:optimal_velocity} the optimal velocity is the conditional mean of the target, and because the target $\bm{\epsilon}-\bm{\mu}(\bm{c})$ is determined by $\bm{x}_t$ and $\bm{c}$, this conditional mean equals the constant $\bm{x}_1-\bm{\mu}(\bm{c})$ along the path. The probability-flow ODE in~\eqref{eq:pf_ode} therefore has a constant right-hand side, so its exact solution is the straight line in~\eqref{eq:straight_line}. A constant-velocity ODE is integrated without error by one Euler step, which gives the zero-discretization claim. The reverse-time SDE that shares the same time marginals adds an independent Brownian increment at every step. Its per-step noise variance does not vanish as the conditional concentrates, so it must remove injected noise across many steps and retains residual spread around $\bm{\mu}(\bm{c})$ at any finite step count.
\end{IEEEproof}

Proposition~\ref{prop:straight_line} characterizes the idealized deterministic limit, and the observed condition refines this picture, because it supplies only a binary building map and coarse geometry, without exact scatterer positions or material parameters. The NLOS APS therefore retains a small but genuine residual uncertainty, and this uncertainty places the task in the perception-distortion regime, where flow matching handles both ends of the regime with one model. When the conditional is sharp, as in LOS, the ODE maps essentially all of the prior to the single mode, and the model then behaves like a sharp regressor without the $\ell_2$ averaging that over-smooths. When residual multi-modality remains, as in NLOS, the noise-to-sample map represents the several modes and yields calibrated diversity, whereas a pure distortion-minimizing regressor instead returns the mean of those modes. This modal average is the over-smoothing failure that flattens the dynamic range and erases the angular structure. The analysis therefore predicts sharper conditional modes than a noise-injecting baseline at equal backbone and step budget, and hence lower $\mathcal{W}_1$ and a stronger beam sweep, a pair of predictions that the experiments in Sec.~\ref{sec:experiments} confirm directly.

\subsection{Per-Metric Estimator Portfolio}
\label{sec:method_deploy}

One trained model serves every deployment need through the choice of readout, and each metric receives its Bayes-appropriate estimator. For distributional fidelity the sampling readout is used unchanged, since the transported samples carry the conditional law itself. For a point estimate the receiver needs a minimum mean-squared-error (MMSE) estimate, which is the posterior mean of the APS. RadioDiff-v2 forms it by averaging $K$ conditional samples,
\begin{align}
  \hat{\bm{x}}_0^{\mathrm{MMSE}} = \frac{1}{K}\sum_{k=1}^{K}
  \bm{x}_0^{(k)}, \qquad \bm{x}_0^{(k)} \sim p_{w}(\bm{x}_0 \mid \bm{c}),
  \label{eq:posterior_mean}
\end{align}
which approximates $\mathbb{E}[\bm{x}_0\mid\bm{c}]$ as $K$ grows. The clean-signal head supplies a second, cheaper route to the same quantity. The reconstruction term in the training loss makes this head a conditional-mean regressor operating on a noisy probe. Evaluating it at the noise end of the path therefore requires no target information,
\begin{align}
  \hat{\bm{x}}_0^{\mathrm{head}} = \frac{1}{R}\sum_{r=1}^{R}
  \mathrm{head}_{x_0}\!\big(t_{\mathrm{r}}\,\bm{\epsilon}_r,\; t_{\mathrm{r}},\; \bm{c}\big),
  \qquad \bm{\epsilon}_r \sim \mathcal{N}(\bm{0},\bm{I}),
  \label{eq:x0head}
\end{align}
with $t_{\mathrm{r}}$ close to one and $R$ antithetic draws. A convex blend $\lambda\,\hat{\bm{x}}_0^{\mathrm{head}} + (1-\lambda)\,\hat{\bm{x}}_0^{\mathrm{MMSE}}$ serves the per-bin metrics, with $\lambda$ chosen per metric. For beam selection the base station needs directions to probe, and the correct decision rule depends on the probing budget $B$. For a small budget the portfolio scores each codebook beam by its gain on the expected spectrum and picks the top beams. For a larger budget it selects the set greedily,
\begin{align}
  \mathcal{B}_{j} = \mathcal{B}_{j-1} \cup \Big\{ \arg\max_{b}\;
  \mathbb{E}_{k}\big[ \max\!\big( u_{b}^{(k)},\,
  \max_{b' \in \mathcal{B}_{j-1}} u_{b'}^{(k)} \big) \big] \Big\},
  \label{eq:beam_set}
\end{align}
where $u_{b}^{(k)}$ is the gain of beam $b$ on sample $k$. Each greedy step adds the beam with the largest expected improvement of the realized best gain. The selected set therefore covers distinct plausible arrivals instead of repeating the dominant one. This objective is monotone submodular, which makes the greedy set near-optimal. The portfolio hyper-parameters include $t_{\mathrm{r}}$, $R$, $K$, $w$, $\lambda$, and the per-budget beam rule. All are selected on held-out validation links from the training environments and frozen before any test evaluation. A regressor that outputs only $\bm{\mu}(\bm{c})$ supplies a point estimate but no distribution, no diverse beam set, and no likelihood, so it cannot host such a portfolio.

\begin{table*}[ht]
\centering
\renewcommand{\arraystretch}{1.4}
\setlength{\tabcolsep}{2.9pt}
\caption{Comparison with baselines on the zero-shot test set over $20$ unseen environments. Distributional and per-bin metrics use a $50{,}000$-link environment-stratified subset, the NLOS beam block uses $20{,}000$ NLOS links, and the LOS block uses the full test; every method is scored on identical links. Portfolio readout hyper-parameters are selected on validation data only and frozen before testing. All values in dB or degrees as indicated; arrows give the optimisation direction. Best is in \textbf{bold} and second-best is \underline{underlined}. The final row reports the RadioDiff-v2 portfolio against the strongest baseline in each column.}
\label{tab:main}
{\footnotesize
\resizebox{0.95\linewidth}{!}{
\begin{tabular}{@{}l|c|c|c|c|c|c|c|c|c|c|c@{}}
\hline
 & Distributional & \multicolumn{3}{c|}{Per-bin distortion} & \multicolumn{4}{c|}{NLOS beam sweep (dB)} & \multicolumn{3}{c@{}}{LOS communication} \\
Method & $\mathcal{W}_1$ (dB) $\downarrow$ & NMSE $\downarrow$ & PSNR $\uparrow$ & SSIM $\uparrow$ & @1 $\downarrow$ & @4 $\downarrow$ & @8 $\downarrow$ & AngSp.$^{\circ}$ $\downarrow$ & $L_{\mathrm{BF}@1}$ $\downarrow$ & AngSp.$^{\circ}$ $\downarrow$ & DoA$^{\circ}$ $\downarrow$ \\
\hline
RME-GAN~\cite{zhang2023rme}          & 2.53          & 0.217             & 22.64          & 0.584          & 12.47            & \underline{7.28} & \underline{4.60} & 26.2       & 0.23          & 2.5          & 4.5 \\
MS-Areg~\cite{huang2026map2aps}      & 5.06          & \underline{0.199} & \underline{23.28} & \underline{0.668} & \underline{11.57} & 8.10        & 5.85          & 27.6          & \underline{0.05} & 3.0       & \underline{2.5} \\
RadioDiff~\cite{wang2024radiodiff}   & 1.97          & 0.312             & 22.99          & 0.611          & 14.23            & 10.98         & 8.50          & \underline{10.8} & 0.08       & \underline{2.3} & 2.9 \\
RadioDiff-v2 sampling                & \underline{0.39} & 0.351          & 22.68          & 0.576          & 13.95            & 10.90         & 8.60          & 11.1          & \textbf{0.02} & \textbf{1.9} & \textbf{1.9} \\
\textbf{RadioDiff-v2 portfolio (Ours)} & \textbf{0.39}  & \textbf{0.184} & \textbf{25.12} & \textbf{0.709} & \textbf{11.33}  & \textbf{4.29} & \textbf{2.43} & \textbf{7.60} & \textbf{0.02} & \textbf{1.9} & \textbf{1.9} \\
\hline
\rowcolor{gray!15}
Gain vs. best baseline & \textcolor{red}{$-80.2\%$} & \textcolor{red}{$-7.5\%$} & \textcolor{red}{$+1.84$} & \textcolor{red}{$+6.1\%$} & \textcolor{red}{$-2.1\%$} & \textcolor{red}{$-41.1\%$} & \textcolor{red}{$-47.2\%$} & \textcolor{red}{$-29.6\%$} & \textcolor{red}{$-60.0\%$} & \textcolor{red}{$-17.4\%$} & \textcolor{red}{$-24.0\%$} \\
\hline
\end{tabular}}
}
\end{table*}
\section{Experiments}
\label{sec:experiments}
\subsection{Generative-MAP Localization}
\label{sec:method_loc}

The learned conditional density doubles as a Bayesian likelihood for localization, which a regressor cannot supply. Given an observed APS $\bm{y}$ measured at an unknown receiver, each candidate position $\bm{p}$ defines a condition $\bm{c}(\bm{p})$. The known building map and the geometry to the serving transmitter induce this condition. Scoring $\bm{y}$ under that condition by the flow-matching loss yields the negative log-likelihood up to a constant. The maximum-a-posteriori (MAP) estimate is the position that minimizes this score,
\begin{align}
  \mathcal{L}_{\mathrm{FM}}(\bm{y}\mid\bm{c}(\bm{p}))
  &= \mathbb{E}_{\bm{\epsilon},t}
  \big\lVert \bm{v}_\phi(\bm{y}_t,t,\bm{c}(\bm{p})) - (\bm{\epsilon}-\bm{y}) \big\rVert_2^2,
  \label{eq:loc_score} \\
  \bm{p}^\star &= \arg\min_{\bm{p}}\;
  \mathcal{L}_{\mathrm{FM}}\big(\bm{y}\mid\bm{c}(\bm{p})\big),
  \label{eq:loc_single}
\end{align}
with $\bm{y}_t=(1-t)\bm{y}+t\bm{\epsilon}$ and the true position excluded from the search. The score in~\eqref{eq:loc_score} rewards positions whose conditional density explains the observed multipath pattern, so it uses the full angular structure rather than a single summary. When the receiver is heard by $M$ base stations, the per-station observations are conditionally independent given the position, so their likelihoods multiply and their log-likelihoods add. The joint estimate sums the per-station scores,
\begin{align}
  \bm{p}^\star = \arg\min_{\bm{p}}\;
  \sum_{m=1}^{M} \mathcal{L}_{\mathrm{FM}}\big(\bm{y}_m \mid \bm{c}_m(\bm{p})\big),
  \label{eq:loc_multi}
\end{align}
where $\bm{y}_m$ and $\bm{c}_m(\bm{p})$ are the observation and the condition for the $m$-th station. The per-station scores trace likelihood ridges along their respective bearings, and the ridges intersect at the true position, so the joint objective triangulates. The single-observation problem is weakly determined, so the decisive gains come from this fusion and require multi-station coverage. The localization accuracy of both forms is reported in Sec.~\ref{sec:experiments}.

\subsection{Experimental Setup}
\label{sec:exp-setup}
We evaluated RadioDiff-v2 on a large-scale angular radio map corpus of 99 environments~\cite{huang2026map2aps}, each containing on the order of $50{,}000$ transmitter-to-receiver links. Every link supplies a dB-domain APS with $N_\theta = 180$ angular bins as the prediction target. The condition consists of the receiver-centred angle-aligned polar building-occupancy map, the transmitter and receiver geometry, and a LOS flag whose angle is withheld. We adopted a strict zero-shot protocol with 79 training environments and 20 disjoint test environments. Training used $3{,}864{,}000$ links, and validation used $60{,}000$ in-distribution links. Every evaluation below uses zero-shot test links from the unseen environments. Each table states its environment-stratified evaluation subset, and every method is scored on identical links. All estimator hyper-parameters of the portfolio were selected on the validation links only and frozen before any test evaluation. About $15$ to $20$ percent of receivers are heard by at least three transmitters, which yields the multi-base-station setting used for localization.

We compared against four baselines that span the relevant method families, and RadioDiff is the prior denoising diffusion probabilistic model with a convolutional U-Net backbone~\cite{wang2024radiodiff,ho2020denoising,ronneberger2015u}. RME-GAN is a conditional generative adversarial network trained with an $\ell_1$ and adversarial loss~\cite{zhang2023rme,creswell2018generative}, and MS-Areg is the multi-scale regression baseline from the Map2APS benchmark, trained with the mean-squared error~\cite{huang2026map2aps}. COST231 is the analytical path-loss model whose angular analogue steers to the geometric receiver-to-transmitter bearing~\cite{cost231,wahl2005dominant}. All learned baselines used the same condition and the same zero-shot split as RadioDiff-v2, so every reported difference reflects the model rather than the data.

The evaluation uses metrics from two regimes that the perception-distortion view keeps separate~\cite{blau2018perception}, and distributional fidelity is measured by the Wasserstein-1 distance $\mathcal{W}_1$ in the dB domain~\cite{peyre2019computational}. Per-bin distortion is measured by the NMSE, the peak signal-to-noise ratio (PSNR), and the structural similarity index (SSIM)~\cite{wang2004image}. Communication quality is measured by the beam-sweep loss $L_{\mathrm{sweep}@B}$, defined as the oracle-best beamforming gain minus the realized gain over $B$ proposed beams. The angular-spread error, the direction-of-arrival (DoA) error, and the single-beam loss $L_{\mathrm{BF}@1}$ complete the communication metrics. Localization quality is measured by the median pixel error on the zero-shot test set, and lower values are better for every metric except PSNR and SSIM.

\subsection{Distributional Fidelity and Per-Bin Distortion}
\label{sec:exp-fidelity}
Table~\ref{tab:main} reports the full comparison across distributional fidelity, per-bin distortion, and both communication regimes. This subsection reads its distributional and per-bin columns. The sampling readout attains a dB-domain Wasserstein-1 of $0.39$, against $1.97$ for the prior diffusion baseline and $2.53$ and $5.06$ for the adversarial and regression baselines. This margin corresponds to a $5.1\times$ to $13.0\times$ improvement. The portfolio readout of the same model wins the per-bin columns as well, with an NMSE of $0.184$ against the regressor's $0.199$, a PSNR of $25.12$ against $23.28$, and an SSIM of $0.709$ against $0.668$. The qualitative comparison in Fig.~\ref{fig:aps_qualitative} explains the distributional gap. RadioDiff-v2 consistently places its peak at the correct angle, whereas the regressor and the adversarial baseline stay broad and noisy near their own peak with no clean dominant lobe, even once each curve is referenced to its own maximum. The flow-matching model recovers the multipath angular structure that the downstream tasks consume, and its clean-signal readout supplies the conditional mean when a point estimate is the goal.

The two RadioDiff-v2 rows of Table~\ref{tab:main} make the perception-distortion structure of the task explicit, because the trade-off binds a single readout, not a single model~\cite{blau2018perception}. The sampling readout matches the conditional distribution and therefore pays a per-bin penalty, exactly as the theory predicts. The clean-signal readout of the same network is instead a conditional-mean regressor, and its blend with the posterior mean beats the strongest regressor on NMSE by $7.5$ percent, on PSNR by $1.84$ dB, and on SSIM by $6.1$ percent. The regressors cannot make the reverse move, because a conditional mean carries no distribution to recover, and their Wasserstein-1 stays an order of magnitude worse. A fairness check gave the prior diffusion baseline the same clean-signal readout, and the readout improves that baseline to an NMSE of $0.191$, which still trails the portfolio. The more accurate conditional distribution therefore yields the more accurate conditional mean, and one model now occupies both ends of the perception-distortion frontier.

\FloatBarrier

\subsection{NLOS Communication}
\label{sec:exp-nlos}
The NLOS beam-sweep columns of Table~\ref{tab:main} address the regime where the conditional law over candidate beams is genuinely multi-modal, and here the decision rule matters as much as the model. The Bayes-optimal portfolio selection of~\eqref{eq:beam_set} attains an eight-beam sweep loss of $2.43$ dB, against $4.60$ dB for the adversarial baseline, $5.85$ dB for the regression baseline, and $8.50$ dB for the prior diffusion baseline. The four-beam sweep shows the same ordering at $4.29$ dB against $7.28$ dB for the closest baseline. The single-probe sweep is the budget where a sharp point estimate helps most, and the expected-spectrum rule still wins it at $11.33$ dB against the regressor's $11.57$ dB. The greedy set also repairs the angular-spread error to $7.60$ degrees, below the prior diffusion baseline at $10.8$ and far below the regressors near $27$. The naive rule that reads one beam per sample reaches only $13.95$ dB at a single probe. The gain from the Bayes rules is thus decision-theoretic and costs no retraining.

Table~\ref{tab:nlos_point} isolates the NLOS point estimate, where the clean-signal readout of~\eqref{eq:x0head} reaches a single-beam loss of $11.23$ dB and a DoA error of $47.4$ degrees. The expected-spectrum rule over $K=16$ samples attains $11.34$ dB and $46.9$ degrees, and both run ahead of the strongest regressor, which sits at $11.57$ dB and $48.1$ degrees. The expected-spectrum estimator keeps the angular-spread error at $7.60$ degrees at the same time, so the accuracy does not come from over-smoothing, and the prior diffusion baseline trails on every column. The advantage is not generic to generative sampling; it follows from the more accurate conditional law and the matched readouts. One model consequently serves both roles, the diverse sweep through its samples and the accurate point estimate through its portfolio.

\begin{table}[ht]
\centering
\renewcommand{\arraystretch}{1.35}
\setlength{\tabcolsep}{10pt}
\caption{NLOS point estimate on $20{,}000$ zero-shot NLOS links, identical for every method. Best is in \textbf{bold}, second-best is \underline{underlined}. The Gain row is the RadioDiff-v2 expected-spectrum estimator against the strongest baseline in each column.}
\label{tab:nlos_point}
\resizebox{\columnwidth}{!}{
\begin{tabular}{@{}l|c|c|c@{}}
\hline
Estimator & DoA err $\downarrow$ & $L_{\mathrm{BF}@1}$ (dB) $\downarrow$ & AngSpread err $\downarrow$ \\
\hline
RME-GAN~\cite{zhang2023rme}          & 52.0          & 12.47         & 26.2 \\
MS-Areg~\cite{huang2026map2aps}       & 48.1          & 11.57         & 27.6 \\
RadioDiff~\cite{wang2024radiodiff}    & 59.2          & 14.23         & \underline{10.8} \\
RadioDiff-v2 clean-signal readout     & \underline{47.4} & \textbf{11.23} & 26.6 \\
\textbf{RadioDiff-v2 expected spectrum ($K{=}16$)} & \textbf{46.9} & \underline{11.34} & \textbf{7.60} \\
\hline
\rowcolor{gray!15}
Gain over best baseline & \textcolor{red}{$-2.5\%$} & \textcolor{red}{$-2.0\%$} & \textcolor{red}{$-29.6\%$} \\
\hline
\end{tabular}}
\end{table}

\FloatBarrier

\subsection{LOS Communication and the Geometric Baseline}
\label{sec:exp-los}
The LOS columns of Table~\ref{tab:main} report the regime where the conditional law is concentrated and the task is well posed. RadioDiff-v2 is best on every LOS metric, with a single-beam loss of $0.02$ dB, an angular-spread error of $1.9$ degrees, and a DoA error of $1.9$ degrees. All methods sit near the optimum here, because the dominant path is visible and the spectrum is nearly a single lobe. The margin over the regressors is therefore small in absolute terms, yet RadioDiff-v2 is the tightest on all three columns. The result confirms that flow matching behaves like a sharp regressor in the concentrated regime without paying the $\ell_2$ averaging penalty that over-smooths the spectrum.

The geometric COST231 baseline explains why learning is essential only in NLOS, and it was evaluated on $60{,}000$ zero-shot links against the true-APS oracle beam, with its angular analogue steering a single beam to the geometric receiver-to-transmitter bearing~\cite{cost231,wahl2005dominant}. In LOS the geometric bearing is the optimal beam, with a single-beam loss of $0.00$ dB, a DoA error of $0.6$ degrees, and a beam identical to the oracle on $99.0$ percent of links. In NLOS the same rule fails, because it points at the blocked direct path, and its single-beam loss grows to $12.12$ dB, its DoA error to $48.0$ degrees, and its beam matches the oracle on only $4.6$ percent of links. The RadioDiff-v2 portfolio lowers the NLOS single-beam loss to $11.34$ dB, below the geometric $12.12$ dB, and its greedy sweep drives the eight-beam loss down to $2.43$ dB. The geometric baseline thus marks the trivial LOS optimum and the NLOS failure that motivates a learned conditional model.

\subsection{Localization}
\label{sec:exp-loc}
Table~\ref{tab:localization} reports localization, the second contribution. The correct way to localize with a generative model is generative-MAP estimation. We score a candidate position by the flow-matching loss of the observed APS under that position's condition and select the minimum with the query position excluded~\cite{yapar2022locunet}. Single-station generative-MAP attains a median error of $62.6$ pixels, which beats the oracle true-APS fingerprint nearest-neighbour at $78.0$ pixels by about $20$ percent. It also improves on the naive use of RadioDiff-v2 as a plain fingerprint at $83.6$ pixels by about $25$ percent. The generative likelihood draws on the learned conditional density and the known building map. The regressors lack this capability entirely, since they expose no conditional likelihood.

\begin{table}[ht]
\centering
\renewcommand{\arraystretch}{1.35}
\setlength{\tabcolsep}{12pt}
\caption{Localization median pixel error on the zero-shot test set. Best is in \textbf{bold}, second-best is \underline{underlined}. The upper block is single-station, and the lower block is multi-base-station joint generative-MAP.}
\label{tab:localization}
\resizebox{\columnwidth}{!}{
\begin{tabular}{@{}l|c@{}}
\hline
Single-station method & Median error (px) $\downarrow$ \\
\hline
RadioDiff-v2 as fingerprint kNN      & 83.6 \\
Oracle true-APS fingerprint kNN      & \underline{78.0} \\
\textbf{RadioDiff-v2 generative-MAP} & \textbf{62.6} \\
\hline
\rowcolor{gray!15}
Gain over oracle fingerprint & \textcolor{red}{$-19.7\%$} \\
\hline
\multicolumn{2}{@{}c@{}}{} \\
\hline
\multicolumn{2}{@{}l@{}}{Multi-station fusion (median px), $M = 1 \,/\, 2 \,/\, 3 \,/\, 4 \,/\, 5 \,/\, 6$} \\
\hline
MS-Areg map $+$ kNN fusion & 68.5 / 56.3 / 46.2 / 42.6 / 41.6 / 41.4 \\
\textbf{RadioDiff-v2 joint generative-MAP} & 74.7 / \textbf{54.1} / \textbf{31.0} / \textbf{20.6} / \textbf{19.2} / \textbf{17.1} \\
\hline
\rowcolor{gray!15}
Gain ($M{=}4$) over regressor-map fusion & \textcolor{red}{$-51.6\%$} \\
\hline
\end{tabular}}
\end{table}

\begin{figure}[ht]
\centering
\begin{subfigure}[b]{0.32\columnwidth}
\centering
\includegraphics[width=\linewidth]{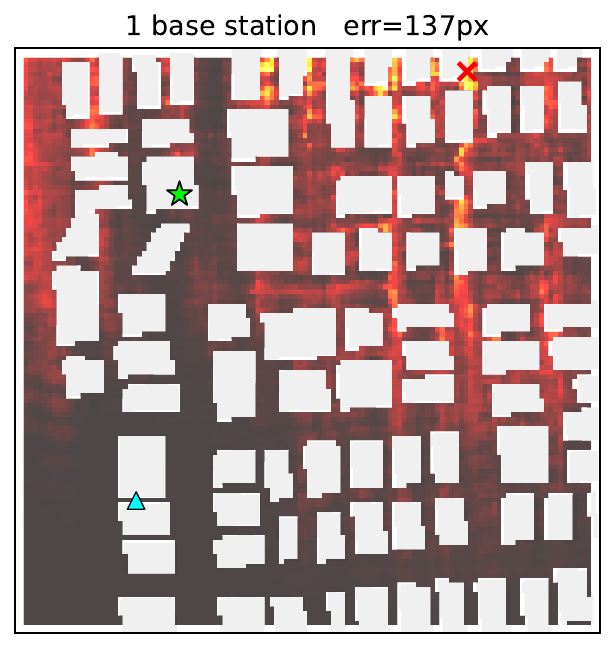}
\caption{One station.}
\label{fig:loc_map_1bs}
\end{subfigure}%
\hfill
\begin{subfigure}[b]{0.32\columnwidth}
\centering
\includegraphics[width=\linewidth]{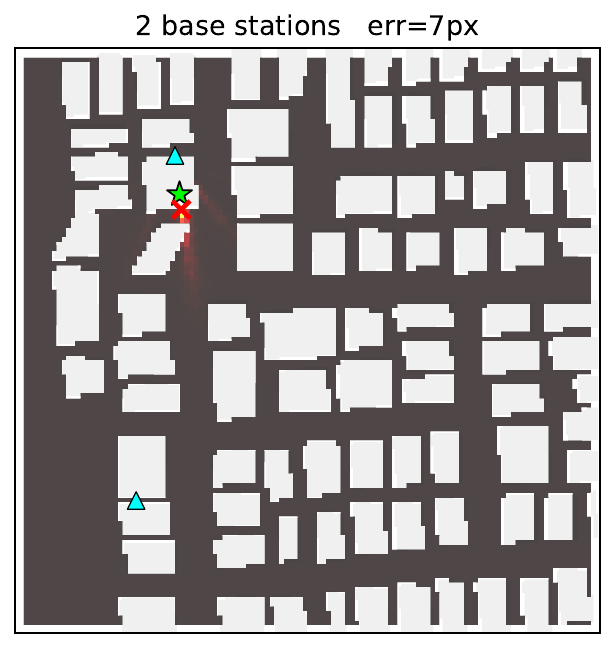}
\caption{Two stations.}
\label{fig:loc_map_2bs}
\end{subfigure}%
\hfill
\begin{subfigure}[b]{0.32\columnwidth}
\centering
\includegraphics[width=\linewidth]{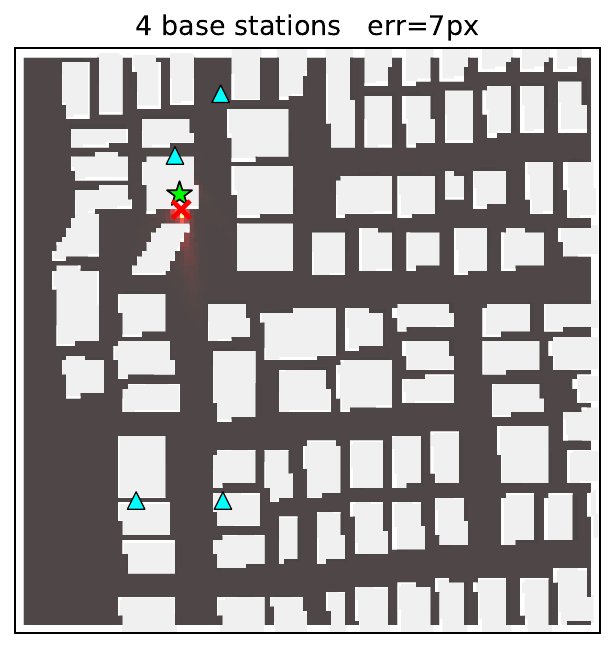}
\caption{Four stations.}
\label{fig:loc_map_4bs}
\end{subfigure}
\caption{Multi-base-station likelihood maps over space for one query with one, two, and four base stations. Each station contributes a likelihood ridge along its bearing, and as stations are added the ridges intersect and the posterior sharpens onto the true position.}
\label{fig:localization_map}
\end{figure}

\FloatBarrier

The lower block of Table~\ref{tab:localization} reports multi-base-station fusion on $577$ fixed queries over all $20$ test environments, and Fig.~\ref{fig:localization_map} shows the posterior sharpening onto the true position. Independent per-station likelihoods multiply, so we sum the per-station flow-matching losses with uniform weights, and the median error falls monotonically from $74.7$ pixels with one station to $20.6$ at four and $17.1$ at six. The fair alternative gave the regressor the same information, fusing nearest-neighbour matches against its predicted maps from the same $M$ stations. This fusion saturates near $41$ pixels from three stations onward, because fingerprint matching carries no likelihood geometry and the extra stations add no triangulation. The generative likelihood instead keeps improving, since each station contributes a ridge along its bearing and the ridges intersect at the true position, and the advantage reaches $2.4\times$ at six stations. Two scope boundaries remain: single-observation localization is weakly determined, and at one station the regressor fusion is slightly ahead on this multi-station query subset. The broader single-station protocol in the upper block still favours generative-MAP by $20$ percent, and multi-station coverage is also required, which the census places at $15$ to $20$ percent of receivers for three or more transmitters. The recommended recipe is therefore the plain generative-MAP for a single station and the uniform-sum fusion across stations.

\subsection{Ablation: Inference-Axis Sampler Selection}
\label{sec:exp-ablation}
Table~\ref{tab:sampler} ablates the classifier-free guidance scale on a test subset and shows that the sampler is an inference-axis design choice~\cite{ho2022classifier}. The guidance scale that minimizes per-bin NMSE is not the scale that minimizes the Wasserstein-1 distance. A greedy NMSE choice at $w = 2$ reaches an NMSE of $0.306$ but a Wasserstein-1 of $4.47$. The Wasserstein-optimal choice at $w = 1$ instead attains a Wasserstein-1 of $0.63$ at a small NMSE cost of $0.334$. Selecting the sampler by the Wasserstein-1 objective thus improves distributional fidelity by about $7\times$ on this subset with no retraining. We deploy the sampler at $w = 1$ under this objective, which is consistent with the sampling-readout fidelity of $0.39$ and the beam-sweep advantage in Table~\ref{tab:main}. The frontier numbers are reported on their own subset and are kept separate from the headline results to avoid cross-quoting.

\begin{table}[ht]
\centering
\renewcommand{\arraystretch}{1.35}
\setlength{\tabcolsep}{12pt}
\caption{Classifier-free guidance frontier on a zero-shot test subset. The NMSE-optimal and Wasserstein-optimal operating points differ, which motivates the Wasserstein-1 sampler selection. Best is in \textbf{bold}, second-best is \underline{underlined}.}
\label{tab:sampler}
\resizebox{\columnwidth}{!}{
\begin{tabular}{@{}c|c|c|c@{}}
\hline
Guidance $w$ & NMSE $\downarrow$ & $\mathcal{W}_1$ (dB) $\downarrow$ & sweep@8 $\downarrow$ \\
\hline
$2.0$ (NMSE-greedy)        & \textbf{0.306} & 4.47          & 5.18 \\
$1.5$                      & \underline{0.303} & \underline{2.89} & \underline{4.64} \\
$1.0$ ($\mathcal{W}_1$-optimal) & 0.334     & \textbf{0.63} & \textbf{4.14} \\
\hline
\rowcolor{gray!15}
Gain ($w{=}1$) over $w{=}2$ & \textcolor{red}{$+9.2\%$} & \textcolor{red}{$-85.9\%$} & \textcolor{red}{$-20.1\%$} \\
\hline
\end{tabular}}
\end{table}

\FloatBarrier

\section{Conclusion}
\label{sec:conclusion}
We have reframed angular radio map prediction as a perception-distortion problem and proposed RadioDiff-v2, a dual-branch one-dimensional diffusion transformer trained with flow matching, whose per-metric estimator portfolio reads every deployment quantity from one model. We have shown that this portfolio leads every baseline on every metric of the zero-shot test, combining a thirteen-fold distributional advantage with regression-grade per-bin accuracy, Bayes-optimal beam selection, and multi-station likelihood localization that continues to improve where regressor-map fusion saturates. One zero-shot map therefore serves spectrum reconstruction, multi-beam selection, and receiver localization for environment-aware 6G networks. Future work will extend the conditional density to wideband and elevation-resolved spectra so that frequency-selective and three-dimensional beam management fit within the same generative framework.

\ifCLASSOPTIONcaptionsoff
  \newpage
\fi

\bibliographystyle{IEEEtran}
\bibliography{ref}

@article{huang2026map2aps,
  title={Map2APS: A Physically Grounded Benchmark for Direct Angle Power Spectrum Prediction from Urban Geometry},
  author={Huang, Junxi and Wang, Xiucheng and Cheng, Nan and Wang, Kailong and Sun, Ruijin and Yin, Zhisheng},
  journal={arXiv preprint arXiv:2605.14989},
  year={2026}
}

@article{zeng2021toward,
  title={Toward environment-aware {6G} communications via channel knowledge map},
  author={Zeng, Yong and Xu, Xiaoli},
  journal={IEEE Wireless Commun.},
  volume={28},
  number={3},
  pages={84--91},
  year={2021}
}

@inproceedings{song2020denoising,
  title={Denoising diffusion implicit models},
  author={Song, Jiaming and Meng, Chenlin and Ermon, Stefano},
  booktitle={Proc. Int. Conf. Learn. Represent. (ICLR)},
  pages = {1-12},
  year={2021}
}

@inproceedings{LDM,
      title={High-Resolution Image Synthesis with Latent Diffusion Models},
      author={Robin Rombach and Andreas Blattmann and Dominik Lorenz and Patrick Esser and Björn Ommer},
      booktitle={Proc. IEEE/CVF Conf. Comput. Vis. Pattern Recognit. (CVPR)},
      pages={10674--10685},
      year={2022}
}

@article{zhang2024generative,
  title={Generative AI on SpectrumNet: An Open Benchmark of Multiband 3D Radio Maps},
  author={Zhang, Shuhang and Jiang, Shuai and Lin, Wanjie and Fang, Zheng and Liu, Kangjun and Zhang, Hongliang and Chen, Ke},
  journal={{IEEE} Trans. Cognit. Commun. Networking},
  volume={11},
  number={2},
  pages={886--901},
  year={2025},
  publisher={IEEE}
}

@article{wang2020indoor,
  title={Indoor radio map construction and localization with deep Gaussian processes},
  author={Wang, Xiangyu and Wang, Xuyu and Mao, Shiwen and Zhang, Jian and Periaswamy, Senthilkumar CG and Patton, Justin},
  journal={{IEEE} Internet Things J.},
  volume={7},
  number={11},
  pages={11238--11249},
  year={2020},
  publisher={IEEE}
}

@inproceedings{wahl2005dominant,
  title={Dominant path prediction model for urban scenarios},
  author={Wahl, Ren{\'e} and W{\"o}lfle, Gerd and Wertz, Philipp and Wildbolz, Pascal and Landstorfer, Friedrich},
  booktitle={14th IST mobile and wireless communications summit},
  pages={1--5},
  year={2005}
}

@article{6g,
    author = {Cheng, Nan and Chen, Fangjiong and Chen, Wen and Cheng, Zhimi and Yang, Qinghai and Li, Changle and Shen, Xuemin},
    title =  {{6G} omni-scenario on-demand services provisioning: vision, technology and prospect(in Chinese)},
    journal = {Sci Sin Inform},
    year = {2024},
    volume={54},
    number={5},
    pages={1025--1054}
}

@article{ho2020denoising,
  title={Denoising diffusion probabilistic models},
  author={Ho, Jonathan and Jain, Ajay and Abbeel, Pieter},
  journal={Advances in neural information processing systems (NeurIPS)},
  volume={33},
  pages={6840--6851},
  year={2020}
}

@article{liu2023exploiting,
  title={Exploiting Radio Fingerprints for Simultaneous Localization and Mapping},
  author={Liu, Ran and Lau, Billy Pik Lik and Ismail, Khairuldanial and Chathuranga, Achala and Yuen, Chau and Yang, Simon X and Guan, Yong Liang and Mao, Shiwen and Tan, U-Xuan},
  journal={{IEEE} Pervasive Comput.},
  year={2023},
  volume={22},
  number={3},
  pages={38--46},
  publisher={IEEE}
}

@article{sun2024generative,
  title={Generative AI for Deep Reinforcement Learning: Framework, Analysis, and Use Cases},
  author={Sun, Geng and Xie, Wenwen and Niyato, Dusit and Mei, Fang and Kang, Jiawen and Du, Hongyang and Mao, Shiwen},
  journal={IEEE Wireless Commun.},
  volume={32},
  number={3},
  pages={186--195},
  year={2025},
  publisher={IEEE}
}

@article{zhang2023rme,
  title={{RME}-{GAN}: A learning framework for radio map estimation based on conditional generative adversarial network},
  author={Zhang, Songyang and Wijesinghe, Achintha and Ding, Zhi},
  journal={{IEEE} Internet Things J.},
  year={2023},
  volume={10},
  number={20},
  pages={18016-18027},
  publisher={IEEE}
}

@inproceedings{li2022radionet,
  title={{RadioNet}: Robust deep-learning based radio fingerprinting},
  author={Li, Haipeng and Gupta, Kaustubh and Wang, Chenggang and Ghose, Nirnimesh and Wang, Boyang},
  booktitle={Proceedings of the 2022 IEEE Conference on Communications and Network Security (CNS)},
  pages={190--198},
  year={2022}
}

@inproceedings{ronneberger2015u,
  title={U-net: Convolutional networks for biomedical image segmentation},
  author={Ronneberger, Olaf and Fischer, Philipp and Brox, Thomas},
  booktitle={Medical image computing and computer-assisted intervention--MICCAI 2015: 18th international conference, Munich, Germany, October 5-9, 2015, proceedings, part III 18},
  pages={234--241},
  year={2015},
  organization={Springer}
}

@article{vaswani2017attention,
  title={Attention is all you need},
  author={Vaswani, Ashish and Shazeer, Noam and Parmar, Niki and Uszkoreit, Jakob and Jones, Llion and Gomez, Aidan N and Kaiser, {\L}ukasz and Polosukhin, Illia},
  journal={Advances in neural information processing systems (NeurIPS)},
  volume={30},
  year={2017}
}

@inproceedings{nichol2021glide,
  title={Glide: Towards photorealistic image generation and editing with text-guided diffusion models},
  author={Nichol, Alex and Dhariwal, Prafulla and Ramesh, Aditya and Shyam, Pranav and Mishkin, Pamela and McGrew, Bob and Sutskever, Ilya and Chen, Mark},
  booktitle={Proc. Int. Conf. Mach. Learn. (ICML)},
  series={PMLR},
  volume={162},
  pages={16784--16804},
  year={2022}
}

@inproceedings{wu2024ckmimagenet,
  title={CKMImageNet: A Comprehensive Dataset to Enable Channel Knowledge Map Construction via Computer Vision},
  author={Wu, Di and Wu, Zijian and Qiu, Yuelong and Fu, Shen and Zeng, Yong},
  booktitle={2024 IEEE/CIC International Conference on Communications in China (ICCC Workshops)},
  pages={114--119},
  year={2024},
  organization={IEEE}
}

@article{creswell2018generative,
  title={Generative adversarial networks: An overview},
  author={Creswell, Antonia and White, Tom and Dumoulin, Vincent and Arulkumaran, Kai and Sengupta, Biswa and Bharath, Anil A},
  journal={IEEE Signal Process. Mag.},
  volume={35},
  number={1},
  pages={53--65},
  year={2018},
  publisher={IEEE}
}

@article{croitoru2023diffusion,
  title={Diffusion models in vision: A survey},
  author={Croitoru, Florinel-Alin and Hondru, Vlad and Ionescu, Radu Tudor and Shah, Mubarak},
  journal={IEEE Trans. Pattern Anal. Mach},
  volume={45},
  number={9},
  pages={10850--10869},
  year={2023},
  publisher={IEEE}
}

@article{li2024radiogat,
  title={RadioGAT: A joint model-based and data-driven framework for multi-band radiomap reconstruction via graph attention networks},
  author={Li, Xiaojie and Zhang, Songyang and Li, Hang and Li, Xiaoyang and Xu, Lexi and Xu, Haigao and Mei, Hui and Zhu, Guangxu and Qi, Nan and Xiao, Ming},
  journal={IEEE Trans. Wireless Commun.},
  year={2024},
  volume={23},
  number={11},
  pages={17777--17792},
  publisher={IEEE}
}

@article{jaensch2024radio,
  title={Radio Map Estimation--An Open Dataset with Directive Transmitter Antennas and Initial Experiments},
  author={Jaensch, Fabian and Caire, Giuseppe and Demir, Beg{\"u}m},
  journal={arXiv preprint arXiv:2402.00878},
  year={2024}
}

@article{wang2025radiodiffk,
  title={RadioDiff-k$^{2}$: Helmholtz Equation Informed Generative Diffusion Model for Multi-Path Aware Radio Map Construction},
  author={Wang, Xiucheng and Zhang, Qiming and Cheng, Nan and Sun, Ruijin and Li, Zan and Cui, Shuguang and Shen, Xuemin},
  journal={IEEE J. Sel. Areas Commun.},
  volume={44},
  pages={2318--2333},
  year={2026},
  doi={10.1109/JSAC.2025.3641105},
  publisher={IEEE}
}

@article{alkhateeb2019deepmimo,
  title={DeepMIMO: A generic deep learning dataset for millimeter wave and massive MIMO applications},
  author={Alkhateeb, Ahmed},
  journal={arXiv preprint arXiv:1902.06435},
  year={2019}
}

@article{wang2025radiodiffinverse,
  title={RadioDiff-Inverse: Diffusion Enhanced Bayesian Inverse Estimation for ISAC Radio Map Construction},
  author={Wang, Xiucheng and Fang, Zhongsheng and Cheng, Nan and Sun, Ruijin and Zhou, Haibo and Su, Zhou and Li, Zan and Shen, Xuemin},
  journal={IEEE Trans. Wireless Commun.},
  volume={25},
  pages={14611--14626},
  year={2026},
  doi={10.1109/TWC.2026.3677479},
  publisher={IEEE}
}

@article{huang2023decoupled,
  title={Decoupled diffusion models with explicit transition probability},
  author={Huang, Yuhang and Qin, Zheng and Liu, Xinwang and Xu, Kai},
  journal={arXiv preprint arXiv:2306.13720},
  year={2023}
}

@article{wang2004image,
  title={Image quality assessment: from error visibility to structural similarity},
  author={Wang, Zhou and Bovik, Alan C and Sheikh, Hamid R and Simoncelli, Eero P},
  journal={{IEEE} Trans. Image Processing},
  volume={13},
  number={4},
  pages={600--612},
  year={2004},
  publisher={IEEE}
}

@inproceedings{song2020score,
  title={Score-based generative modeling through stochastic differential equations},
  author={Song, Yang and Sohl-Dickstein, Jascha and Kingma, Diederik P and Kumar, Abhishek and Ermon, Stefano and Poole, Ben},
  booktitle={Proc. Int. Conf. Learn. Represent. (ICLR)},
  year={2021}
}

@article{zhou2017electromagnetic,
  title={Electromagnetic scattering laws in Weyl systems},
  author={Zhou, Ming and Ying, Lei and Lu, Ling and Shi, Lei and Zi, Jian and Yu, Zongfu},
  journal={Nature Commun.},
  volume={8},
  number={1},
  pages={1388},
  year={2017},
  publisher={Nature Publishing Group UK London}
}

@article{deschamps1972ray,
  title={Ray techniques in electromagnetics},
  author={Deschamps, Georges A},
  journal={Proc. {IEEE}},
  volume={60},
  number={9},
  pages={1022--1035},
  year={1972},
  publisher={IEEE}
}

@article{kang2025confidence,
  title={Confidence-Regulated Generative Diffusion Models for Reliable AI Agent Migration in Vehicular Metaverses},
  author={Kang, Yingkai and Kang, Jiawen and Wen, Jinbo and Zhang, Tao and Yang, Zhaohui and Niyato, Dusit and Zhang, Yan},
  journal={arXiv preprint arXiv:2505.12710},
  year={2025}
}

@techreport{docomo20165g,
  title={5G Channel Model for bands up to 100 GHz},
  author={Docomo, NTT and others},
  year={2016},
  institution={Technical report}
}

@inproceedings{yapar2022locunet,
  title={Locunet: Fast urban positioning using radio maps and deep learning},
  author={Yapar, {\c{C}}a{\u{g}}kan and Levie, Ron and Kutyniok, Gitta and Caire, Giuseppe},
  booktitle={ICASSP 2022-2022 IEEE International Conference on Acoustics, Speech and Signal Processing (ICASSP)},
  pages={4063--4067},
  year={2022},
  organization={IEEE}
}

@article{levie2021radiounet,
  title={{RadioUNet}: Fast radio map estimation with convolutional neural networks},
  author={Levie, Ron and Yapar, {\c{C}}a{\u{g}}kan and Kutyniok, Gitta and Caire, Giuseppe},
  journal={ {IEEE} Trans. Wireless Commun.},
  volume={20},
  number={6},
  pages={4001--4015},
  year={2021},
  publisher={IEEE}
}

@article{wang2024radiodiff,
  author={Wang, Xiucheng and Tao, Keda and Cheng, Nan and Yin, Zhisheng and Li, Zan and Zhang, Yuan and Shen, Xuemin},
  journal={IEEE Trans. Cognit. Commun. Networking}, 
  title={RadioDiff: An Effective Generative Diffusion Model for Sampling-Free Dynamic Radio Map Construction}, 
  year={2025},
  volume={11},
  number={2},
  pages={738-750}
}

@inproceedings{lipman2023flow,
  title={Flow Matching for Generative Modeling},
  author={Lipman, Yaron and Chen, Ricky T. Q. and Ben-Hamu, Heli and Nickel, Maximilian and Le, Matt},
  booktitle={Proc. Int. Conf. Learn. Represent. (ICLR)},
  year={2023}
}

@inproceedings{liu2023flow,
  title={Flow Straight and Fast: Learning to Generate and Transfer Data with Rectified Flow},
  author={Liu, Xingchao and Gong, Chengyue and Liu, Qiang},
  booktitle={Proc. Int. Conf. Learn. Represent. (ICLR)},
  year={2023}
}

@inproceedings{peebles2023scalable,
  title={Scalable Diffusion Models with Transformers},
  author={Peebles, William and Xie, Saining},
  booktitle={Proc. IEEE/CVF Int. Conf. Comput. Vis. (ICCV)},
  pages={4172--4182},
  year={2023}
}

@misc{ho2022classifier,
  title={Classifier-Free Diffusion Guidance},
  author={Ho, Jonathan and Salimans, Tim},
  year={2022},
  eprint={2207.12598},
  archivePrefix={arXiv},
  primaryClass={cs.LG}
}

@inproceedings{blau2018perception,
  title={The Perception-Distortion Tradeoff},
  author={Blau, Yochai and Michaeli, Tomer},
  booktitle={Proc. IEEE/CVF Conf. Comput. Vis. Pattern Recognit. (CVPR)},
  pages={6228--6237},
  year={2018}
}

@article{peyre2019computational,
  title={Computational Optimal Transport: With Applications to Data Science},
  author={Peyr{\'e}, Gabriel and Cuturi, Marco},
  journal={Found. Trends Mach. Learn.},
  volume={11},
  number={5-6},
  pages={355--607},
  year={2019},
  doi={10.1561/2200000073}
}

@techreport{cost231,
  title={Digital Mobile Radio Towards Future Generation Systems---{COST} Action 231 Final Report},
  author={Damosso, E. and Correia, L. M.},
  institution={European Commission},
  number={EUR 18957},
  address={Brussels, Belgium},
  year={1999}
}

@article{wang2026radiodiff3d,
  title={{RadioDiff}-3{D}: A 3{D}$\times$3{D} Radio Map Dataset and Generative Diffusion Based Benchmark for 6{G} Environment-Aware Communication},
  author={Wang, Xiucheng and Zhang, Qiming and Cheng, Nan and Chen, Junting and Zhang, Zezhong and Li, Zan and Cui, Shuguang and Shen, Xuemin},
  journal={IEEE Trans. Netw. Sci. Eng.},
  volume={13},
  pages={3773--3789},
  year={2026}
}

@article{wang2025iradiodiff,
  title={{iRadioDiff}: Physics-Informed Diffusion Model for Indoor Radio Map Construction and Localization},
  author={Wang, Xiucheng and Yuan, Tingwei and Cao, Yang and Cheng, Nan and Sun, Ruijin and Zhuang, Weihua},
  journal={arXiv preprint arXiv:2511.20015},
  year={2025}
}

@article{wang2026tutorial,
  title={A Tutorial on Learning-Based Radio Map Construction: Data, Paradigms, and Physics-Awareness},
  author={Wang, Xiucheng and Pan, Yuhao and Cheng, Nan and Yapar, {\c{C}}a{\u{g}}kan and Sun, Ruijin and Yin, Zhisheng and Zhou, Conghao and Xu, Wenchao and Zhang, Yuxiang and Zhang, Jianhua and Cui, Shuguang and Shen, Xuemin},
  journal={arXiv preprint arXiv:2603.17499},
  year={2026}
}

@article{wang2026radiodiffflux,
  title={{RadioDiff}-Flux: Efficient Radio Map Construction via Generative Denoise Diffusion Model Trajectory Midpoint Reuse},
  author={Wang, Xiucheng and Zheng, Peilin and Jia, Honggang and Cheng, Nan and Sun, Ruijin and Zhou, Conghao and Shen, Xuemin},
  journal={IEEE Trans. Cognit. Commun. Networking},
  volume={12},
  pages={4882--4895},
  doi={10.1109/TCCN.2025.3641513},
  year={2026}
}

@article{wang2026radiodifffs,
  title={{RadioDiff}-{FS}: Physics-Informed Manifold Alignment in Few-Shot Diffusion Models for High-Fidelity Radio Map Construction},
  author={Wang, Xiucheng and Guo, Zixuan and Cheng, Nan},
  journal={arXiv preprint arXiv:2603.18865},
  year={2026}
}

\end{document}